\documentclass[10pt,conference]{IEEEtran}

\makeatletter
\def\ps@headings{%
\def\@oddhead{\mbox{}\scriptsize\rightmark \hfil \thepage}%
\def\@evenhead{\scriptsize\thepage \hfil \leftmark\mbox{}}%
\def\@oddfoot{}%
\def\@evenfoot{}}
\usepackage{times}
\usepackage{graphicx}
\usepackage{url}
\usepackage{color}
\usepackage{multirow}
\usepackage{amsmath}
\usepackage{amssymb}

\newtheorem{theorem}{Theorem}

\newtheorem{lemma}{Lemma}
\newtheorem{definition}{Definition}

\newcommand{\bA}{\mathbf{A}}

\newcommand{\BF}{\mathbb{F}}
\newcommand{\BS}{\mathbb{S}}
\newcommand{\SF}{{\cal F}}
\newcommand{\SG}{{\cal G}}
\newcommand{\SP}{{\cal P}}
\newcommand{\SQ}{{\cal Q}}
\newcommand{\SR}{{\cal R}}
\newcommand{\SU}{{\cal U}}
\newcommand{\SV}{{\cal V}}
\newcommand{\SX}{{\cal X}}
\newcommand{\SY}{{\cal Y}}
\newcommand{\SZ}{{\cal Z}}

\begin{document}

%----------------------------------------------------------------------------
\title{\huge Analysis and Construction of Functional Regenerating Codes with
	Uncoded Repair for Distributed Storage Systems}

\author{Yuchong Hu$^\dagger$, Patrick P. C. Lee$^\ddagger$, and Kenneth W. Shum$^\dagger$\\
$^\dagger$Institute of Network Coding, The Chinese University of Hong Kong\\
$^\ddagger$Department of Computer Science and Engineering, The Chinese University of Hong Kong\\
\{ychu,wkshum\}@inc.cuhk.edu.hk, pclee@cse.cuhk.edu.hk}

\maketitle

%\thispagestyle{empty}

%----------------------------------------------------------------------------
\begin{abstract}
Modern distributed storage systems apply redundancy coding techniques to
stored data.  One form of redundancy is based on regenerating codes, which can
minimize the repair bandwidth, i.e., the amount of data transferred when
repairing a failed storage node.  Existing regenerating codes mainly require
surviving storage nodes encode data during repair.  In this paper, we study
{\em functional minimum storage regenerating (FMSR)} codes, which enable 
{\em uncoded} repair without the encoding requirement in surviving nodes, 
while preserving the
minimum repair bandwidth guarantees and also minimizing disk reads.  Under
double-fault tolerance settings, we formally prove the existence of FMSR
codes, and provide a deterministic FMSR code construction that can
significantly speed up the repair process.  We further implement and evaluate
our deterministic FMSR codes to show the benefits.  Our work is built atop 
a practical cloud storage system that implements FMSR codes, and we provide
theoretical validation to justify the practicality of FMSR codes. 
\end{abstract}

\section{Introduction}
\label{sec:introduction}

We have witnessed the wide deployment of storage systems in
Internet-wide distributed settings, such as peer-to-peer storage (e.g.,
\cite{Kubiatowicz00,Bhagwan04,Chun06,Weatherspoon07}) and cloud storage
(e.g., GFS \cite{Ghemawat03} and Azure \cite{Calder11}), in which data 
is striped over multiple storage nodes in a networked environment.  For data
availability, a storage system must keep user data for a long period of time
and allow users to access their data on demand.  However, storage nodes are
often deployed in commodity machines and prone to failures \cite{Ghemawat03}.
It is thus important for a storage system to ensure data availability in
practical deployment.

One way to ensure data availability is to store redundant data over multiple
storage nodes.  Redundancy can be achieved via {\em maximum distance separable
(MDS)} codes such as Reed-Solomon codes \cite{reed60}, whose idea is that 
even if any subset of nodes fail, the original data remains accessible from
the remaining surviving nodes.  In general, Reed-Solomon codes have
significantly less redundancy overhead than simple replication of data under
the same fault tolerance requirement. 

%One way to ensure data availability is to store data with redundancy.
%Redundancy can be introduced via {\em maximum distance separable (MDS)} codes
%defined by two parameters $n$ and $k$ ($< n$). An ($n$, $k$) MDS code divides
%a file of size $M$ into $k$ fragments of size $M/k$ each.  It then encodes the
%fragments into $n$ fragments such that any $k$ out of $n$ fragments suffice to
%recover the original file.  By storing the $n$ fragments over $n$ different
%nodes, a storage system can tolerate at most $n-k$ concurrent node failures.
%One example of MDS codes is the Reed-Solomon erasure codes \cite{reed60},
%which can incur significantly less redundancy overhead than replication of
%data under the same fault tolerance level. 

When a storage node fails, it is necessary to recover the lost data of the
failed node to preserve the required level of fault tolerance.  
{\em Regenerating codes} \cite{dimakis10} have been proposed to minimize the
{\em repair bandwidth}, which defines the amount of data traffic transferred
in the repair process.  Regenerating codes are built on network coding
\cite{ahlswede00}, such that to repair a failed node, existing surviving nodes
encode their own stored data and send the encoded data to the new node, which
then reconstructs the lost data.  It is shown that regenerating codes use less
repair bandwidth than Reed-Solomon codes, given the same storage overhead and
fault tolerance requirements. 

However, there are challenges of deploying regenerating codes in practice.
First, most regenerating code constructions (e.g.,
\cite{wu09a,Cullina09,Cadambe10,suh11b,rashmi11,shah12})  %shum13
require storage nodes to encode stored data during repair. This may
not be feasible for some storage devices that merely
provide the basic I/O functionalities without any encoding capabilities. 
More importantly, even if storage nodes have encoding capabilities, they must
first read all available data from disk and combine the data into encoded form
before transmitting encoded data for repair.  This leads to high disk reads, 
which may degrade the actual repair performance. 
%Authors in~\cite{Shah12b} propose a class of regenerating codes that perform
%{\em uncoded} repair through mere transfer of data (called the 
%{\em repair-by-transfer} model) without encoding operations in storage nodes.
%However, their proposed codes do not satisfy the MDS property (see
%Section~\ref{sec:related}) for details.  

On the applied side, a cloud storage system NCCloud \cite{Hu12} proposes and
implements {\em functional minimum storage regenerating (FMSR)} codes, which
have several key properties: (i) FMSR codes preserve the fault tolerance of
MDS codes and have the same redundancy overhead as MDS codes for a given
fault tolerance; (ii) FMSR codes preserve the benefits of network
coding as they minimize the repair bandwidth (e.g., the repair bandwidth
saving compared to RAID-6 codes is up to 50\% \cite{Hu12}); and (iii) FMSR
codes use {\em uncoded} repair without
requiring encoding of surviving nodes during repair, and this can minimize
disk reads as the amount of data read from disk is the same as that being
transferred.  FMSR codes are designed as {\em non-systematic} codes as they do
not keep the original uncoded data as their systematic counterparts, but
instead store only linear combinations of original data called 
{\em parity chunks}.  Each round of repair regenerates new parity
chunks for the new node and ensures that the fault tolerance level is
maintained.  A trade-off of FMSR codes is that the whole encoded file must be
decoded first if parts of a file are accessed.  Nevertheless, FMSR codes are
suited to long-term archival applications, since data backups are
rarely read and it is common to restore the whole file rather than file
parts. 

While FMSR codes have been experimented on real-life cloud testbeds, there
remain open issues regarding whether FMSR codes exist and how they are
deterministically constructed.  In particular, given that new parity chunks
are regenerated in each round of repair, we need to ensure that such chunks
preserve the fault tolerance of MDS codes after multiple rounds of repair.
Thus, the key motivation of this work is to {\em provide
theoretical foundations for the practicality of FMSR codes}. 

In this paper, we conduct formal analysis on the existence of FMSR codes and
provide a deterministic construction for FMSR codes, with an objective of
theoretically validating the practicality of FMSR codes in distributed
storage systems.  We focus on the double-fault tolerance setting (i.e., at
most two node failures can be tolerated) as in conventional RAID-6
codes \cite{raid6}.  Note that double-fault tolerance is by default used in
practical cloud storage systems such as GFS \cite{Ghemawat03} and Azure
\cite{Calder11}.  Our contributions are three-fold.
\begin{itemize}
\item
We formally prove the existence of FMSR codes with uncoded repair, such that 
the fault tolerance of MDS codes is preserved after any number of rounds of
repair.  
\item
We provide a deterministic FMSR code construction, such that the repair can
deterministically specify (i) the chunks to be read from surviving nodes and
(ii) the encoding coefficients used to regenerate new chunks.  This
significantly speeds up the repair time compared to the random FMSR code
construction used in NCCloud \cite{Hu12}. 
\item 
We build and evaluate our deterministic FMSR codes, and show that the chunk
selection and regeneration during repair can be finished within less than
one second. 
\end{itemize}

The rest of the paper proceeds as follows.
Section~\ref{sec:related} reviews related work. 
Section~\ref{sec:model} characterizes the system model of FMSR codes and
formulates the problems. 
Section~\ref{sec:existence} formally proves the existence of FMSR codes. 
Section~\ref{sec:deterministic} provides a deterministic FMSR code
construction. 
Section~\ref{sec:evaluation} presents evaluation results. 
Section~\ref{sec:conclusions} concludes the paper.

\section{Background and Related Work}
\label{sec:related}

%Examples of MDS erasure codes for %storage applications (See Plank's tutorial
%in~\cite{Plank05}) include %Reed-Solomon codes~\cite{reed60}, CRC
%codes~\cite{Blomer95}, %EVENODD~\cite{Blaum95}, X-code~\cite{Xu99},
%RDP~\cite{Corbett04}, %STAR-code~\cite{Huang08}, etc.

%LDPC codes~\cite{Gallager63}, Tornado codes~\cite{Luby01}, LT codes~\cite{Luby02}, Raptor codes~\cite{Raptor06}: these codes are asymptotically MDS, so they are not included here.

Dimakis et al. \cite{dimakis10} first propose {\em regenerating codes} based
on {\em network coding} \cite{ahlswede00} for distributed storage systems. It
is shown that when repairing a single failed storage node, regenerating codes
use less repair bandwidth than conventional Reed-Solomon codes \cite{reed60}
by transmitting encoded data from the surviving nodes to a new node. Also,
\cite{dimakis10} gives an optimal tradeoff spectrum between storage cost and
repair bandwidth and identifies two extreme points.  One extreme point refers
to the {\em minimum storage regenerating (MSR)} codes, in which each node
stores the minimum amount of data as in Reed-Solomon codes.  Another extreme
point is the {\em minimum bandwidth regenerating (MBR)} codes, which allow
each node to store more data than in conventional Reed-Solomon codes to minimize
the repair bandwidth.  In this work, we focus on the MSR codes, so that we can
fairly compare with conventional Reed-Solomon codes under the same storage
overhead.

As shown in \cite{dimakis10,Yunnanwu10,hu10}, the MSR point is achievable
under {\em functional-repair}, which means that the repaired data may not be
the same as the lost data while still maintaining the same fault tolerance
level.  However, the corresponding coding schemes perform random linear coding
in surviving nodes and do not provide explicit construction.  Then there are
extensive studies (e.g.,
\cite{wu09a,Cullina09,Cadambe10,suh11b,rashmi11,shah12}) on the
\emph{exact-repair} MSR (EMSR) codes, in which the data reconstructed is
identical to the lost data.

Most EMSR codes require storage nodes encode stored data during repair.
Authors in \cite{Rouayheb10,Shah12b} propose regenerating codes that
eliminate encoding of storage nodes during repair. We call it
{\em uncoded repair} \cite{Rouayheb10}, or {\em repair-by-transfer}
\cite{Shah12b}.  However, their constructions belong to MBR codes.  EMSR code
constructions based on uncoded repair have been proposed in
\cite{tamo11,wang11}.  The EMSR code in \cite{tamo11} has the uncoded repair
property for systematic nodes that store original data chunks but not for the
parity nodes that store encoded chunks, while that in \cite{wang11} has
the uncoded repair property for both systematic and parity nodes. However,
the code construction in \cite{wang11} requires the total number of data
chunks being stored increase exponentially with the number of systematic
nodes.  This increases the number of chunk accesses, and limits its
application in practical storage systems.

Several studies (e.g., \cite{wang10,xiang11,khan12}) %,zhu12a,zhu12b}) 
propose uncoded repair schemes that minimize disk reads for XOR-based
erasure codes.  Their solutions are built on existing code constructions. 
In general, they do not achieve the global minimum point.  
%instead of proposing new storage codes as in \cite{tamo11,wang11}.  Thus in
%among all possible code constructions.

A recent applied work \cite{Hu12} builds a network-coding-based cloud storage
system called NCCloud.  The authors build and evaluate {\em functional
MSR (FMSR)} codes, which minimize the repair bandwidth using uncoded repair.
%We formally describe the FMSR code design in Section~\ref{sec:model}.  
Later in \cite{Shum12}, the correctness of FMSR codes is analyzed for a
special case of two systematic nodes.  In this paper, we generalize the
analysis, and also provide a deterministic code construction, for more
systematic nodes.

\section{System Model for FMSR Codes}
\label{sec:model}

\subsection{Basics of FMSR Codes}

We first describe the basics of FMSR codes, which are used by NCCloud
\cite{Hu12} to store files over multiple independent storage nodes. Each
node could be a disk device, a storage server, or a cloud storage provider.
NCCloud motivates using FMSR codes to provide fault-tolerant, long-term
archival storage using multiple clouds, so as to save the monetary cost in
migrating data between cloud providers during repair.  FMSR codes have
three design properties, which we elaborate below. 

%and $\bp_i$ be the column vector $[p_{i,1} \ p_{i,2} \ \ldots \
%p_{i,n-k}]^T$.

{\bf Property 1: FMSR codes preserve the fault tolerance and storage
efficiency of MDS codes.}  MDS codes are defined by two parameters $n$ and $k$
($k<n$). An ($n$, $k$)-MDS code divides a file of size $M$ into $k$ pieces of
size $M/k$ each, and encodes them into $n$ pieces such that any $k$ out of $n$
encoded pieces suffice to recover the original file.  By storing the $n$
encoded pieces over $n$ nodes, a storage system can tolerate at most $n-k$
node failures.  An example of MDS codes is Reed-Solomon codes \cite{reed60}.

Figure~\ref{fig:fmsr} shows the FMSR codes for a special case $n=4$
and $k=2$. To store a file of size $M$ units, an ($n$, $k$)-FMSR code splits
the file evenly into $k(n-k)$ {\em native chunks}, say $F_1, F_2, \ldots,
F_{k(n-k)}$, and encodes them into $n(n-k)$ {\em parity chunks} of size
$\frac{M}{k(n-k)}$ each.  Each $l^{th}$ parity chunk is formed by a linear
combination of the $k(n-k)$ native chunks, i.e., $\sum_{m=1}^{k(n-k)}
\alpha_{l,m} F_m$ for some encoding coefficients $\alpha_{l,m}$.  All encoding
coefficients and arithmetic are operated over a finite field $\BF_q$ of
size $q$.  We store the $n(n-k)$ parity chunks on $n$ nodes, each keeping
$n-k$ parity chunks.  Note that native chunks need not be stored. The original
file can be restored by decoding $k(n-k)$ parity chunks of any $k$ nodes, where
decoding can be done by inverting an encoding matrix \cite{Plank97}.  Let
$P_{i,j}$ be the $j^{th}$ parity chunk stored on node $i$, where $i=1, 2,
	\ldots, n$ and $j=1, \ldots, n-k$.

\begin{figure}[t]
\centering
\includegraphics[width=0.9\linewidth]{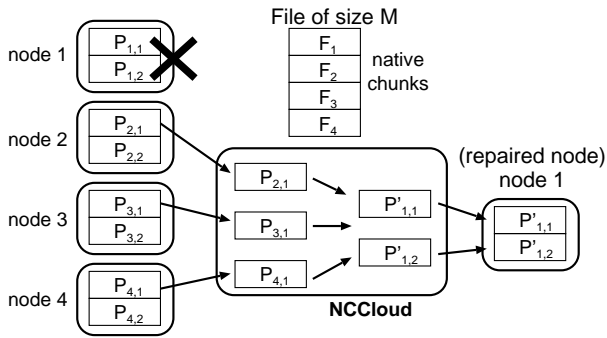}
\caption{FMSR codes with $n=4$ and $k=2$.}
\label{fig:fmsr}
\vspace{-1em}
\end{figure}

{\bf Property 2: FMSR codes minimize the repair bandwidth.}  If a node
fails, we must reconstruct the lost data of the failed node to preserve fault
tolerance.  The conventional repair of Reed-Solomon codes reads $k$ pieces
from any $k$ surviving nodes to restore the original file (by the design of
MDS codes).  Clearly, the amount of data read is the file size $M$.  FMSR
codes seek to read less than $M$ units of data to reconstruct the lost data.
We define {\em repair bandwidth} as the amount of data read from surviving
nodes during repair.  FMSR codes are designed to match the minimum storage
point of regenerating codes when repairing a node failure \cite{dimakis10},
while having each node store $M/k$ units of data as in Reed-Solomon codes. 
To repair a failed node in FMSR codes, each surviving node transfers data of
size $\frac{M}{k(n-k)}$ units as in \cite{dimakis10}, or equivalently, a size
of one parity chunk.  In a special case of $n=4$ and $k=2$ (see
Figure~\ref{fig:fmsr}), the repair bandwidth is 0.75$M$, i.e., 25\% less than
that of conventional repair of Reed-Solomon codes.  In general, the repair
bandwidth of FMSR codes for $k=n-2$ is $\frac{M(n-1)}{2(n-2)}$, and its saving
compared to RAID-6 codes \cite{raid6} (which are also double-fault tolerant)
is up to 50\% if $n$ is large \cite{Hu12}.  

{\bf Property 3: FMSR codes use {\em uncoded} repair.} During repair, each
surviving node under FMSR codes transfers one parity chunk, without any
encoding operations.  This also minimizes the amount of data read from disk.
Suppose we have a failed node $l$ (e.g., $l=1$ in
Figure~\ref{fig:fmsr}). Then we read one parity chunk denoted by $P_{i,f(i)}$
from each surviving node $i$, where $1\le i\le n$ and $i\ne l$, and $f(.)$
denotes some function that specifies which chunk to be read from a
surviving node.  Then we encode the $n-1$ parity chunks into $n-k$
linearly independent parity chunks $P'_{l,1}, P'_{l,2}, \ldots P'_{l,n-k}$,
which will all be stored in a new node, which becomes the new node $l$ (called
the {\em repaired node}).  Each new parity chunk is generated by:
\begin{equation}
\label{fml:basic}
P'_{l,j} = \sum_{i=1, i\ne l}^n \gamma_{i,j}  P_{i,f(i)}, \textrm{ for }
j=1,2,\ldots n-k,
\end{equation}
where $\gamma_{i,j}$ denotes some coefficient for encoding the collected
parity chunks into new chunks. In Section~\ref{sec:deterministic}, we formally
specify how we choose $f(.)$ and $\gamma_{i,j}$.

\subsection{Formulation of Repair Problem in FMSR Codes}

We formulate the repair problem in FMSR codes based on \cite{Hu12}.  Note
that \cite{Hu12} only gives a high-level description, without formal
definitions and theoretical validations.  Here, we provide a theoretical
framework that formalizes the idea of \cite{Hu12}. 

FMSR codes satisfy the MDS property, as described below.

\begin{definition}
\label{def:mds}
{\em MDS property.} For any subset of $k$ out of $n$ nodes, if the $k(n-k)$
parity chunks from the $k$ nodes can be decoded to the $k(n-k)$ native chunks
of the original file, then the MDS property is satisfied. $\hfill\Box$
\end{definition}

\begin{definition}
{\em Decodability}.  We say that a collection of $k(n-k)$ parity chunks is
{\em decodable} if the parity chunks can be decoded to the original file,
which can be verified by checking if the associated $k(n-k)$ vectors of
encoding coefficients are linearly independent. Note that these $k(n-k)$
chunks may be scattered among $n$ nodes, and need not reside in $k$ nodes.
$\hfill\Box$
\end{definition}

Note that FMSR codes operate on parity chunks.  For simplicity, when we use
the term ``chunk'' in our discussion, we actually refer to a parity chunk.

Since FMSR codes regenerate different chunks in each repair, one design
challenge of FMSR codes is to preserve the MDS property after multiple rounds
of repairs.  We illustrate with an example in Figure~\ref{fig:fmsr}.
Suppose that node~1 fails, and we construct new chunks $P'_{1,1}$ and
$P'_{1,2}$ using $P_{2,1}$, $P_{3,1}$, and $P_{4,1}$ as in
Figure~\ref{fig:fmsr}.  Next, suppose that node~2 fails.  If we construct new
chunks $P'_{2,1}$ and $P'_{2,2}$ using $P'_{1,1}$, $P_{3,1}$, and $P_{4,1}$,
then in the repaired nodes 1 and 2, the chunks $\{P'_{1,1}, P'_{1,2},
P'_{2,1}, P'_{2,2}\}$ are the linear combinations of only three chunks
$P_{2,1}$, $P_{3,1}$, and $P_{4,1}$ instead of four.  So the chunks in the
repaired nodes~1 and 2 are {\em not} decodable, and the MDS property is lost.

%Instead of always selecting the first chunk for repair, which results
%in losing the MDS property as in the above example,
Thus, to preserve the MDS property over multiple rounds of repair, NCCloud
uses a specific implementation of FMSR codes based on random chunk selection,
which we call {\em random FMSR codes}.  NCCloud seeks to {\em completely}
avoid linear dependence in chunk regeneration and hence losing the MDS
property.  Specifically, NCCloud performs the $r^{th}$ (where $r\ge 1$) round
of repair as follows:
\begin{enumerate}
\item[(i)]
It randomly selects a chunk from each surviving node (i.e., $f(.)$
returns a random value), and generates random encoding coefficients to encode
the selected chunks into new chunks (i.e., $\gamma_{i,j}$'s are randomly
chosen).
\item[(ii)]
It then performs {\em two-phase checking}.  In the first phase, it checks if
the MDS property is satisfied with the new chunks generated (i.e., the chunks
of any $k$ out of $n$ nodes remain decodable) after the current $r^{th}$ round
of repair.  In the second phase, it further checks if the MDS property is
still satisfied after the $(r+1)^{th}$ round of repair for any possible node
failure, and this property is called the {\em repair MDS property}.
\item[(iii)]
If both phases are passed, then NCCloud writes the generated chunks to a new
node; otherwise, it repeats (i) and (ii) with another set of random chunks and
random encoding coefficients.
\end{enumerate}

We now formally define the repair MDS property.

%\begin{figure}[t]
%\centering
%\includegraphics[width=\linewidth]{figs/repair_mds.eps}
%\caption{RBC.}
%\label{fig:rpc}
%\end{figure}

\begin{definition}
\label{def:rbc}
{\em Repair-based collections (RBCs).}  An RBC of the $r^{th}$ round of
repair is a collection of $k(n-k)$ chunks that can be obtained after the
$r^{th}$ round of repair by the following
procedure.  (Step~1) We select any $n-1$ out of $n$ nodes. (Step 2) We select
$k-1$ out of the $n-1$ nodes found in Step 1 and collect $n-k$ chunks from
each selected node. (Step 3) We collect one chunk from each of the
non-selected $n-k$ nodes.  Clearly, the number of collected chunks is
$(k-1)(n-k)+(n-k) = k(n-k)$. $\hfill\Box$
\end{definition}

We can easily verify that there are
${{n}\choose{n-1}}{{n-1}\choose{k-1}}(n-k)^{n-k}$ different RBCs. Intuitively,
an RBC refers to a collection of chunks of $k$ nodes after the
$(r+1)^{th}$ round of repair for any possible node failure. For instance, after
repairing node~1 in Figure~\ref{fig:fmsr}, one example RBC is
$\SR = \{P'_{1,1},P_{3,1},P_{3,2},P_{4,1}\}$. This means that we assume:
node~2 is the failed node in the next round of repair; the failed node~2 will
be repaired by chunks $P'_{1,1}$, $P_{3,1}$ (or $P_{3,2}$), and $P_{4,1}$;
and we consider if the chunks of node~2 (after repair) and node~3 are
decodable.  Note that the chunks of node~2 and node~3 are linear combinations
of this RBC $\SR$. 

We assume that when a file is stored, it is first encoded using Reed-Solomon
codes, such that any $k(n-k)$ out of $n(n-k)$ (parity) chunks are decodable.
Note that these $k(n-k)$ chunks may reside in more than $k$ nodes (e.g.,
$P_{1,1}, P_{2,1}, P_{3,1}, P_{4,1}$ in Figure~\ref{fig:fmsr}).  If no repair
is carried out, then we ensure that every possible RBC is decodable.

However, after repairing a node failure, there exist some provably
non-decodable RBCs.  For example, in Figure~\ref{fig:fmsr}, the RBCs
$\{P'_{1,1}, P'_{1,2}, P_{2,1}, P_{3,1}\}$,
$\{P'_{1,1}, P'_{1,2}, P_{2,1}, P_{4,1}\}$, and
$\{P'_{1,1}, P'_{1,2}, P_{3,1}, P_{4,1}\}$ are non-decodable,
since $P'_{1,1}$ and $P'_{1,2}$ are linear combinations of
$P_{2,1}$, $P_{3,1}$, $P_{4,1}$.  Note that these non-decodable RBCs all
contain the chunks of the repaired node~1.  Each of these RBCs is a linear
combination of chunks $P_{2,1}$, $P_{3,1}$, $P_{4,1}$ (i.e., less than four
chunks) in the repair.  Accordingly, we define the following:

\begin{definition}
\label{def:ldc}
{\em Linear Dependent Collection (LDC)}.  Suppose an RBC of the $r^{th}$
round of repair contains the $n-k$ chunks of the repaired node that are
collected in Step~2 (see Definition~\ref{def:rbc}). If and only if every chunk
of this RBC is a linear combination of a set of less than $k(n-k)$ chunks of
the $r^{th}$ round of repair, we call it an LDC of the $r^{th}$ round of
repair.  $\hfill\Box$
\end{definition}

For example, in Figure~\ref{fig:fmsr}, the RBCs $\{P'_{1,1}, P'_{1,2},
P_{2,1}, P_{3,1}\}$, $\{P'_{1,1}, P'_{1,2}, P_{2,1}, P_{4,1}\}$, and
$\{P'_{1,1}, P'_{1,2}, P_{3,1}, P_{4,1}\}$ are the LDCs of the current round
of repair. 

\begin{definition}
\label{def:rmds}
{\em Repair MDS (rMDS) property.}  If all RBCs, after excluding the LDCs, of
the $r^{th}$ round of repair are decodable, then we say the rMDS property is
satisfied.  It means that if every RBC that is a linear combination of exactly
$k(n-k)$ chunks is always decodable, then we say that the rMDS property is
satisfied. 
$\hfill\Box$
\end{definition}

\begin{definition}
\emph{($n$,$k$)-FMSR codes}. An original file is stored in $n$ nodes in the
form of $n(n-k)$ chunks. If these $n(n-k)$ chunks satisfy both the MDS and
rMDS properties, then we say this file is FMSR-encoded.
\end{definition}

\textbf{Summary.} Authors of NCCloud \cite{Hu12} show via simulations that by
checking both the MDS and rMDS properties in each round of repair, FMSR codes
can preserve the MDS property after hundreds of rounds of repair.  Also, if we
check only the MDS property but not the rMDS property, then after some rounds
of repair we cannot regenerate the chunks that preserve the MDS property
within a fixed number of iterations (this is called the {\em bad} repair
\cite{Hu12}).  On the other hand, there is no formal theoretical analysis
showing the need of two-phase checking to preserve the MDS property after 
{\em any} number of rounds of repair.  Also, random FSMR codes repeat
two-phase checking until the valid chunks are regenerated.  This could involve
many iterations and significantly increase the repair time overhead (see
Section~\ref{sec:evaluation}).  In the following sections, we formally provide
the theoretical validation of existence of FMSR codes and the design of
deterministic FMSR codes.

\section{Existence}
\label{sec:existence}

We now prove the existence of FMSR codes.  In this work, we focus on $k=n-2$,
implying that FMSR codes are double-fault tolerant as conventional RAID-6 codes
\cite{raid6}.  Double-fault tolerance has been assumed in practical cloud
storage systems (e.g., GFS \cite{Ghemawat03} and Azure \cite{Calder11}).  Our
goal is to show that FMSR codes always maintain double-fault tolerance (i.e.,
the MDS property is always satisfied with $k=n-2$) after any number of rounds
of {\em uncoded} repair, while the repair bandwidth is kept at the MSR
point.

We first give three lemmas. Lemmas~\ref{lemma_LDC_number} and
\ref{lemma_LDC} provide a guideline of how to choose $n-1$ chunks from $n-1$
surviving nodes (one chunk from each node) to repair a failed node.
Lemma~\ref{lemma_SZ} implies that if the finite field size is large enough,
then we can always find a set of encoding coefficients to regenerate new
chunks for a repaired node so as to maintain the MDS and rMDS properties after
each round of repair.  Finally, we prove Theorem~\ref{theorem_fmsr} for the
existence of FMSR codes.

\begin{lemma}
\label{lemma_LDC_number}
In repair, let $\SF$ be the set of $n-1$ chunks selected from $n-1$ surviving
nodes to regenerate the $n-k$ chunks of the repaired node.  Also, let $\SQ$ be
the set of chunks collected in Step~3 of RBC construction (see
Definition~\ref{def:rbc}).  If an RBC (denoted by $\SR$) containing the $n-k$
chunks of the repaired node is an LDC, then $\SF$ and $\SQ$ must have two or
more common chunks.
\end{lemma}

\begin{proof}
%Let $\SR$ be an RBC that contains the $n-k$ chunks of the repaired node. 
%Each chunk can be seen as a vector with each component in $\BF_q$. 
%Thus, we can define the linear span of $\SR$ as span($\SR$) (i.e., the set of
%all linear combinations of all chunks in $\SR$). 
%We want to prove that if $\SR$ is an LDC, 
%then span($\SR$) has less than $k(n-k)$ chunks (vectors), and 
%then $\SF$ and $\SQ$ have two or more common chunks. 
%
Without loss of generality, let node~1 be the failed node. Let $\SP$ be the 
set of chunks collected in Step~2 of Definition~\ref{def:rbc} excluding
the $n-k$ chunks of the repaired node~1. Thus, 
$\SR = $ $\{P'_{1,1},\ldots,P'_{1,n-k}\} \cup \SP \cup \SQ$. 
As $P'_{1,1},\ldots,P'_{1,n-k}$ are obtained by linearly combining the
chunks in $\SF$, 
%have span($\SR$) $\subseteq$ span($\SF \cup \SP \cup \SQ$).
we infer that $\SR$ contains linear combinations of chunks in 
$\SF \cup \SP \cup \SQ$. 

Since $\SF$ selects one chunk from each of $n-1$ surviving nodes and $\SP$ has
all the chunks from $k-2$ surviving nodes, $\SF$ and $\SP$ have $k-2$
identical chunks, i.e., $|\SF \cap \SP| = k-2$. According to the given
conditions, we can easily have the following equalities: $|{\SF}|=n-1$,
$|\SP|=(k-2)(n-k)$, $|\SQ|= n-k$, $|\SP \cap \SQ| = |{\SF} \cap \SP \cap \SQ|
= 0$.  Finally we can have
$
|\SF \cup \SP \cup \SQ|= |{\SF}| + |\SP| + |\SQ|
- |{\SF} \cap \SP| - |{\SF} \cap \SQ| - |\SP \cap \SQ|
+ |{\SF} \cap \SP \cap \SQ|
= k(n-k) + 1 - |\SF\cap\SQ|. 
$
Since $\SR$ is an LDC, $|\SF \cup \SP \cup \SQ| < k(n-k)$. Hence, 
$|\SF\cap\SQ| \ge 2$.  Lemma~\ref{lemma_LDC_number} holds.
\end{proof}

%Now we define the rMDS property based on Lemma~\ref{lemma_LDC_number}.

\begin{lemma}
\label{lemma_LDC}
Suppose that the rMDS property is satisfied after every $r^{th}$ round of
repair.
Then for any $n-1$ out of $n$ nodes, we can always select one chunk from these
$n-1$ nodes (i.e., a total of $n-1$ chunks) such that any RBC containing
the selected $n-1$ chunks is decodable.
\end{lemma}

\begin{proof}
Without loss of generality, suppose that we construct an RBC $\SR$ by selecting
the chunks from nodes~$2,\ldots,n$ (see Step~1 of Definition~\ref{def:rbc}),
and that $\SG$ be the set of $n-1$ chunks selected from nodes~$2,\ldots,n$
(one chunk from each nodes). We prove the existence of $\SG$ such that if 
$\SR$ contains $\SG$ (i.e., $\SG\subset\SR$), then $\SR$ is decodable. 

If node~1 is the repaired node in the $r^{th}$ round of repair, then $\SR$ is
never an LDC (by Definition~\ref{def:ldc}).  Since the rMDS property is
satisfied by our assumption, $\SR$ is decodable (by
Definition~\ref{def:rmds}). 
   
If node~1 is not the repaired node in the $r^{th}$ round of repair, then
without loss of generality, let node~2 be the repaired node.  By the FMSR
design, the chunks of node~2 are linearly combined by one chunk in each of
nodes~$1, 3, \ldots, n$.  We denote these chunks by 
$\SF = \{P_{1,f(1)},P_{3,f(3)},\ldots,P_{n,f(n)}\}$. Since each node has 
$n-k>1$ chunks, we can construct $\SG = P_{2,g(2)},\ldots,P_{n,g(n)}\}$ such
that $g(i)\neq f(i)$ for $i=3,\ldots,n$ (while $g(2)$ can be randomly picked).
If $\SR$ contains $\SG$, then in Step~3 of RBC construction (see
Definition~\ref{def:rbc}), at least one chunk must be selected from $\SG$.
However,
$\SG$ has no identical chunk with $\SF$.  By Lemma~\ref{lemma_LDC_number},
$\SR$ is not an LDC. Since the rMDS property is satisfied, $\SR$ is decodable. 
\end{proof}

\begin{lemma}
\label{lemma_SZ}
(Schwartz-Zippel Theorem) \cite{SZtheorem}.  Consider a multivariate non-zero
polynomial $h(x_1, \ldots, x_t)$ of total degree $\rho$ over a finite field
$\BF$. Let $\BS$ be a finite subset of $\BF$, and
$\tilde{x}_1, \ldots, \tilde{x}_t$ be the values randomly selected from $\BS$.
Then the probability $\Pr[h(\tilde{x}_1, \ldots, \tilde{x}_t) = 0] \leq
\frac{\rho}{|\BS|}$.
\end{lemma}

\begin{theorem}
\label{theorem_fmsr}
Consider a file encoded using FMSR codes with $k=n-2$.  In the $r^{th}$ ($r\ge
1$) round of uncoded repair of some failed node $j$, the lost chunks are
reconstructed by the random linear combination of $n-1$ chunks selected from
$n-1$ surviving nodes (one chunk from each node).  Then after the repair, the
reconstructed file still satisfies both the MDS and rMDS properties with
probability that can be driven arbitrarily to 1 by increasing the field size
of $\BF_q$.
\end{theorem}

\begin{proof}
We prove by induction on $r$. Initially, we use Reed-Solomon codes to
encode a file into $n(n-k) = 2n$ chunks that satisfy both the MDS and rMDS
properties.  Suppose that after the $r^{th}$ round of repair, both the MDS and
rMDS properties are satisfied (this is our induction hypothesis).

Let $\SU_r=\{P_{1,1},P_{1,2};\ldots;P_{k+2,1}, P_{k+2,2}\}$ be the current set
of chunks after the $r^{th}$ round of repair.  In the $(r+1)^{th}$ round of
repair, without loss of generality, let node~1 be the failed node to repair.
Since $\SU_r$ satisfies the rMDS property, we have the following corollary by
Lemma~\ref{lemma_LDC}.

{\bf Corollary.} There exists a set of $n-1$ chunks, denoted by $\SF =
\{P_{2,f(2)}, \ldots, P_{k+2,f(k+2)}\}$, selected from nodes~$2, \ldots, n$,
such that any RBC containing $\SF$ is decodable. 

We use $\SF$ to repair node~1.  Suppose that the repaired node~1 has the new
chunks $\{P'_{1,1}, P'_{1,2}\}$. Then:
\begin{equation}
\label{theorem_repaired_chunks}
P'_{1,j} = \sum_{i=2}^{k+2} \gamma_{i,j} P_{i,f(i)}, \textrm{ for $j=1,2.$}
\end{equation}

Next we prove that we can always tune $\gamma_{i,j}$ in $\BF_q$ in such a way
that the set of chunks in the $(r+1)^{th}$ round of repair
$\SU_{r+1}=\{P'_{1,1}, P'_{1,2}; \ldots; P_{k+2,1}, P_{k+2,2}\}$ still
satisfies both MDS and rMDS properties. The proof consists of two parts.

\textbf{Part I: ${\cal{U}}_{r+1}$ satisfies the MDS property.} Since $\SU_r$
satisfies the MDS property, we only need to ensure that for any $k-1$
surviving nodes, say for any subset $\{s_1, \dots, s_{k-1}\}$ $\subseteq$
$\{2, \dots, n\}$, all the chunks of nodes $s_1, \dots, s_{k-1}$ and the
repaired node 1 are decodable. Without loss of generality, let
$(s_1, \dots, s_{k-1})$ $=$ $(2,\dots,k)$, and other cases are symmetric.

Let $\SV=\{P_{2,1}, P_{2,2};$$\ldots;$$P_{k,1}, P_{k,2};$$P'_{1,1},
P'_{1,2}\}$ be the set of chunks of nodes~1 to $k$. By
Equation~(\ref{theorem_repaired_chunks}), each chunk of $\SV$ is a linear
combination of a certain RBC, denoted by $\SR = $
$\{P_{2,1}, P_{2,2};\ldots;P_{k,1},P_{k,2};P_{k+1,f(k+1)},P_{k+2,f(k+2)}\}$.
Mathematically, we express as:
\[
\left[
\begin{array}{c}
P_{2,1} \\
P_{2,2} \\
\ldots \\
P_{k,1} \\
P_{k,2} \\
{P'}_{1,1} \\
{P'}_{1,2} \\
\end{array}
\right]
=
\mathbf{A}
\times
\left[
\begin{array}{c}
P_{2,1} \\
P_{2,2} \\
\ldots \\
P_{k,1} \\
P_{k,2} \\
P_{k+1,f(k+1)} \\
P_{k+2,f(k+2)}
\end{array}
\right]\\,
\]
where $\mathbf{A}$ is a $k(n-k)\times k(n-k)$ (i.e., $2k\times 2k$)
encoding matrix given by $\bA = $
\[
\left(
\begin{array}{c@{\ }c@{\ }c@{\ }c}
1,0,&\cdots,&0,0,&0,0\\
0,1,&\cdots,&0,0,&0,0\\
\vdots&  \ddots&  \vdots&\vdots \\
0,0,&\cdots,&1,0,&0,0\\
0,0,&\cdots,&0,1,&0,0\\
\delta_{2,1}\gamma_{2,1},\delta_{2,2}\gamma_{2,1},&\cdots,&\delta_{k,1}\gamma_{k,1},\delta_{k,2}\gamma_{k,1},&\gamma_{k+1,1},\gamma_{k+2,1}\\
\delta_{2,1}\gamma_{2,2},\delta_{2,2}\gamma_{2,2},&\cdots,&\delta_{k,1}\gamma_{k,2},\delta_{k,2}\gamma_{k,2},&\gamma_{k+1,2},\gamma_{k+2,2}
\end{array}
\right)
\]
where $\delta_{i,1}=1$ and $\delta_{i,2}=0$ when $f(i)=1$, and
$\delta_{i,1}=0$ and $\delta_{i,2}=1$ when $f(i)=2$.  Since $\SR$ is an RBC
containing $\SF$, it is
decodable due to Lemma~\ref{lemma_LDC}.  In addition, the determinant
det($\bA$) is a multivariate polynomial in terms of variables $\gamma_{i,j}$.
By Lemma~\ref{lemma_SZ}, the value of det($\bA$) is non-zero, with probability
driven to 1 if we increase the finite field size.  Now since $\SR$ is
decodable and $\bA$ has a full rank, $\SV$ is decodable.  This implies that
$\SU_{r+1}$ satisfies the MDS property.

%Similarly, we can always tune $\gamma_{i,j}$ for another ${\cal{V}}$ that corresponds to any other $\{s_1, \dots, s_{k-1}\}\subseteq \{2, \dots, n\}$ such that ${\cal{V}}$ is decodable. So ${\cal{U}}_{r+1}$ satisfies the MDS Property.

\textbf{Part II: ${\cal{U}}_{r+1}$ satisfies the rMDS property.} By
Definition~\ref{def:rmds}, we need to prove that all the RBCs of
$\SU_{r+1}$ except the LDCs are decodable. By Definition~\ref{def:rbc}, we
consider two cases of RBCs. Without loss of generality, we let node~1 be the
repaired node.

%\emph{Case I}: The new node is selected at Step 1. Because the RBC does not select new node to provide chunks, then the RBC are among $P_2, \dots, P_n$. And because $P_1, \dots, P_n$ satisfies the rMDS Property after the $m^{th}$ repair, the RBC among $P_2, \dots, P_n$ must be able to rebuild the file except the LDCs. By the proof of Lemma~\ref{lemma_LDC_number}, the number of the LDCs among $P_2, \dots, P_n$ is ${{n-2}\choose{k-2}}[(n-k)^{n-k} - (n-k-1)^{n-k} - {{n-k}\choose{1}}(n-k-1)^{n-k-1}]$, which equals to $(n-2)(n-3)/2$ when $n-k=2$.

{\em Case 1}: The repaired node~1 is selected in Step~2. Suppose in
Step~1, an RBC selects any $n-2=k$ surviving nodes, say $\{s_1, \dots,
s_{k}\}\subseteq \{2, \dots, n\}$. Then in Step~2, the RBC further selects any
subset of $k-2$ nodes, say nodes~$s_1,\dots,s_{k-2}$.  The RBC now contains
all the chunks of node~$s_1,\dots,s_{k-2}$ and the repaired node~1. Finally,
in Step~3, the RBC collects two chunks, denoted by $P_{s_{k-1},g(s_{k-1})}$
and $P_{s_k,g(s_k)}$ from the remaining two nodes~$s_{k-1}$ and $s_{k}$,
respectively.  Without loss of generality, let
$(s_1, \dots, s_{k-2}) = (2,\dots,k-1)$ and $(s_{k-1},s_k) = (k,k+1)$.

Denote the RBC by $\SR_1 = \{P_{2,1}, P_{2,2};\ldots;P_{k-1,1},P_{k-1,2};$
$P'_{1,1}, P'_{1,2};P_{k,g(k)},P_{k+1,g(k+1)}\}$.
In addition, by Equation~(\ref{theorem_repaired_chunks}), the chunks of
$\SR_1$ are linear combinations of a set of chunks denoted by
$\SX = \{P_{2,1},P_{2,2};\ldots;P_{k-1,1},P_{k-1,2};$
$P_{k,g(k)},P_{k,f(k)};P_{k+1,g(k+1)},P_{k+1,f(k+1)};P_{k+2,f(k+2)}\}$.

Our goal is to show that if $\SR_1$ is not an LDC, then it is decodable.
By Lemma~\ref{lemma_LDC_number}, we know that if $\SR_1$ is an
LDC, then there are at least two chunks selected in Step 3 that
belong to $\SF = \{P_{2,f(2)}, \ldots, P_{n,f(n)}\}$ (which are used to
regenerate chunks for node~1), or equivalently, $g(k)=f(k)$ and
$g(k+1)=f(k+1)$.  Therefore, to prove that $\SR_1$ except the LDCs is
decodable, it is equivalent to prove that $\SR_1$ is decodable when
(a) $g(k) \neq f(k)$ and $g(k+1) = f(k+1)$, (b) $g(k) = f(k)$ and $g(k+1)
\neq f(k+1)$, or (c) $g(k) \neq f(k)$ and $g(k+1) \neq f(k+1)$.

First consider (a). We can reduce $\SX$ to
$\{P_{2,1},P_{2,2};\ldots;$
$P_{k-1,1},P_{k-1,2};P_{k,1},P_{k,2};P_{k+1,f(k+1)},P_{k+2,f(k+2)}\}$.
The above collection is an RBC containing $\SF$. By our corollary, the
collection is decodable.  Therefore, $\SR_1$ is linear combination of a
decodable collection. Then we can use the similar method in Part~I to prove
that there always exists an assignment of $\gamma_{i,j}$ in a sufficiently
large field such that $\SR_1$ is decodable (by Lemma~\ref{lemma_SZ}).  
The proof of (b) is similar to that of (a) and is thus omitted.

Lastly, let us consider (c). Now, $\SX$ can be written as
$\{P_{2,1},P_{2,2};\ldots;P_{k-1,1},P_{k-1,2};P_{k,1},P_{k,2};P_{k+1,1},P_{k+1,2};$ $P_{k+2,f(k+2)}\}$.
Define $\overline{\SX} = \SX - \{P_{k+2,f_{k+2}}\}$. Note that the MDS
property of $\overline{\SX}$ is satisfied by induction hypothesis. Thus,
$\overline{\SX}$ is decodable, implying that $P_{k+2,f(k+2)}$ can be seen as
a linear combination of $\overline{\SX}$.  Obviously, we can also say that
$\SX$ is formed by linear combinations of $\overline{\SX}$. Therefore, $\SR_1$
is also formed by linear combinations of the decodable collection
$\overline{\SX}$. Based on the above argument, $\SR_1$ is decodable.

%Similarly, we can always tune $\gamma_{i,j}$ for another ${\cal{W}}_1$ that corresponds to any other $\{s_1, \dots, s_{k}\}\subseteq \{2, \dots, n\}$ such that ${\cal{W}}_1$ is decodable.

\emph{Case 2}: The repaired node 1 is selected in Step 3. Suppose in Step~1,
the RBC selects any $n-2=k$ surviving nodes, say
$\{s_1,\dots,s_{k}\}\subseteq \{2,\dots,n\}$. Then in Step~2, the RBC further
selects any subset of $k-1$ nodes, say ${s_1,\dots,s_{k-1}}$ to collects all
the chunks of nodes~$s_1,\dots,s_{k-1}$.  Finally, in Step~3, the RBC collects
two chunks $P'_{1,g(1)}$ and $P_{s_k,g(s_k)}$ from the repaired node 1 and the
last selected node~$s_{k}$, respectively. Without loss of generality, let
$(s_1,\dots,s_{k-1}) = (2,\dots,k)$ and $s_{k}=k+1$.

Denote the RBC by $\SR_2=\{P_{2,1},P_{2,2};\ldots;P_{k,1},P_{k,2};$
$P'_{1,g(1)},P_{k+1,g(k+1)}\}$.  We need to show that if $\SR_2$ is not an LDC,
it is decodable.  Based on Lemma~\ref{lemma_LDC_number}, there is no more than
one identical chunk between $\SF$ and the RBC's chunks collected in Step 3, so
$\SR_2$ is never an LDC. We only need to prove that every possible $\SR_2$ is
decodable.

By Equation~(\ref{theorem_repaired_chunks}), the chunks of $\SR_2$ are linear
combinations of a set of chunks denoted by $\SY =
\{P_{2,1},P_{2,2};\ldots;P_{k,1},P_{k,2};$
$P_{k+1,g(k+1)},P_{k+1,f(k+1)};P_{k+2,f(k+2)}\}$.  Suppose $g(k+1) \neq
f(k+1)$.  Define $\overline{\SY} = \SY - \{P_{k+1,g(k+1)}\}$.
Since $\overline{\SY}$ is an RBC containing $\SF$, by our corollary,
$\overline{\SY}$ is decodable.  Therefore, $P_{k+1,g(k+1)}$ can be seen as a
linear combination of $\overline{\SY}$. Obviously, we can also say $\SY$ is a
linear combination of $\overline{\SY}$. Therefore, $\SR_2$ is also linear
combination of the decodable collection $\overline{\SY}$.  Similar to the
above arguments, $\SR_2$ is decodable. If $g(k+1) = f(k+1)$, the proof is
similar and is thus omitted.

%Similarly, we can always tune $\gamma_{i,j}$ for another ${\cal{W}}_2$ that
%corresponds to any other $\{s_1, \dots, s_{k}\}\subseteq \{2, \dots, n\}$
%such that ${\cal{W}}_2$ is decodable.

Combining Case 1 and Case 2, we deduce that all RBCs excluding the LDCs are
decodable. So $\SU_{r+1}$ satisfies the rMDS property. Therefore,
Theorem~\ref{theorem_fmsr} concludes.
\end{proof}

\section{Deterministic FMSR Codes}
\label{sec:deterministic}

In NCCloud \cite{Hu12}, the repair operation under FMSR codes is accomplished
based on two random processes: (i) using random chunk selection to read chunks
from the surviving nodes and (ii) applying random linear combinations of the
selected chunks to generate new chunks for the repaired node.
Section~\ref{sec:existence} has proved the correctness of the random-based
repair operation by virtue of existence of FMSR codes. On the other hand, a
drawback of the random approach is that it may need to try many iterations to
generate the correct set of chunks that satisfies both the MDS and rMDS
properties.

In this section, we propose a deterministic repair scheme under FMSR codes
($k=n-2$), such that both the chunk selection and linear combination
operations are deterministic. This enables us to significantly speed up the
repair operation.  In our deterministic scheme, we specify which particular
chunk should be read from each surviving node in each round of repair. We also
derive the sufficient conditions on which the encoding coefficients should
satisfy.  To design the deterministic scheme, we first introduce an evolved
repair MDS property.

\begin{definition}
\label{def:ermds}
\emph{Evolved Repair MDS (erMDS) property.} Let $k=n-2$. For any $k+1$ out of
$n$ nodes, if we can always select one specific chunk from each of the $k+1$
nodes such that any RBC containing these selected $k+1$ chunks is decodable,
then we say the code scheme has the erMDS property.  $\hfill\Box$
\end{definition}

From Lemma~\ref{lemma_LDC}, we can see that if the rMDS property is satisfied,
then the erMDS property is also satisfied. Thus, any RBCs satisfying the
rMDS property is a subset of the RBCs satisfying the erMDS property. We use
the erMDS property to construct a deterministic FMSR code. 
%Finally, we discuss and compare the differences of the random and
%deterministic FMSR codes.

%\subsection{Construction of Deterministic FMSR Codes}

To construct deterministic FMSR codes for $k=n-2$, we describe how we store a
file and how we trigger the $r^{th}$ ($r\ge 1$) round of repair for a node
failure.  The correctness of our deterministic FMSR codes is proved in
Appendix. 

\textbf{Storing a file.} We divide a file into $k(n-k) = 2k$
equal-size native chunks, and encode them into $n(n-k) = 2(k+2)$ parity chunks
denoted by $P_{1,1},P_{1,2};$ $\ldots;$ $P_{k+2,1},P_{k+2,2}$ using
Reed-Solomon codes, such that any $2k$ out of $2(k+2)$ chunks are decodable to
the original file.  Each node~$i$ (where $i=1,2,\ldots,k+2$) stores two chunks
$P_{i,1}$ and $P_{i,2}$.  Clearly, the generated parity chunks satisfy the MDS
property (see Definition~\ref{def:mds}), i.e., for any $k$ out of $n$ nodes 
$\{s_1,\ldots,s_k\} \subset \{1,\ldots,k+2\}$, the $2k$ parity chunks
$\{P_{s_1,1},P_{s_1,2};\ldots;P_{s_{k},1},P_{s_{k},2}\}$ are decodable.
In addition, the generated parity chunks also satisfy the erMDS property (see
Definition~\ref{def:ermds}), i.e., for any $k+1$ nodes $s_1,\ldots,s_{k+1}$,
we can always select some specific chunks $P_{s_1,f(s_1)},\ldots,
P_{s_{k+1},f(s_{k+1})}$ such that any RBC containing them is decodable.
Here, we need to find and record such $k+1$ specific chunks for any $k+1$
nodes.  For illustrative purposes, we let $f(s_i)=1$, where
$i=1,2,\ldots,k+1$, so we record the chunks
$\{P_{s_1,1},\ldots,P_{s_{k+1},1}\}$. 

\textbf{The first round of repair.} Suppose without loss of generality that
node~1 fails and is then repaired by two steps.

\emph{Step 1: (Chunk selection)}. We select $k+1$ chunks $P_{2,1},$ $\ldots,$
$P_{k+2,1}$ that are recorded when the file is stored. 

\emph{Step 2: (Coefficient construction)}. For each selected chunk 
$P_{i',1}$ ($i' = 2,\ldots, k+2$), we compute $2k$ coefficients
$\lambda_{i,j}^{(i')}$ ($i=2,\ldots,k+2$, $i\neq i'$, $j=1,2$) which satisfy
\begin{equation}
\label{equ:lambda1}
P_{i',1}=\sum_{i=2,i \neq i'}^{k+2} \sum_{j=1}^2 \lambda_{i,j}^{(i')} P_{i,j}.
%(\lambda_{2,1}^{k+2}, \lambda_{2,2}^{k+2},\ldots, \lambda_{k+1,1}^{k+2},\lambda_{k+1,2}^{k+2})[{\cal{Z}}]^T,
\end{equation}

Each parity chunk is a linear combination of $k(n-k)=2k$ native chunks (see
Section~\ref{sec:model}). By equating the coefficients that are multiplied
with the $2k$ native chunks on both left and right sides of
Equation~(\ref{equ:lambda1}), we obtain $2k$ equations, which allow us to
solve for $\lambda_{i,j}^{(i')}$. 

Next we need to construct the coefficients $\gamma_{i,1}$ and $\gamma_{i,2}$
which satisfy the following inequalities~(\ref{fml:deterministic_condition_1}),~(\ref{fml:deterministic_condition_2}), and~(\ref{fml:deterministic_condition_3}):
\begin{equation}
\label{fml:deterministic_condition_1}
\gamma_{i,1} \gamma_{j,2} \neq \gamma_{i,2} \gamma_{j,1},
\end{equation}
where $i \neq j$ and $i,j=2,3,\ldots,k+2$; 
\begin{equation}
\label{fml:deterministic_condition_2}
\gamma_{i,2} + \gamma_{i',2} \lambda_{i,1}^{(i')} \neq 0,
\end{equation}
where $i \neq i'$ and $i,i'\in\{2,\ldots,k+2\}$; and
\begin{equation}
\label{fml:deterministic_condition_3}
\begin{array}{l}
(\gamma_{i,1} + \gamma_{i'',1}\lambda_{i,1}^{(i'')})
(\gamma_{i',2}+ \gamma_{i'',2}\lambda_{i',1}^{(i'')})
\neq \\
(\gamma_{i',1}+ \gamma_{i'',1}\lambda_{i',1}^{(i'')})
(\gamma_{i,2} + \gamma_{i'',2}\lambda_{i,1}^{(i'')}),
\end{array}
\end{equation}
where $i$, $i'$ and $i''$ are distinct, $i,i',i'' \in \{2,\ldots,k+2\}$. We
can then construct the coefficients $\gamma_{i,1}$ and $\gamma_{i,2}$, and 
by Lemma~\ref{lemma_SZ} the solution exists if the finite field size is
large enough. 
Lastly, we regenerate new chunks $P'_{1,1}$ and $P'_{1,2}$ as follows:
\begin{equation}
\label{fml:deterministic_combination_1}
P'_{1,1}=\gamma_{2,1} P_{2,1}+\gamma_{3,1} P_{3,1}+\ldots+\gamma_{k+2,1} P_{k+2,1},
\end{equation}
\begin{equation}
\label{fml:deterministic_combination_2}
P'_{1,2}=\gamma_{2,2} P_{2,1}+\gamma_{3,2} P_{3,1}+\ldots+\gamma_{k+2,2} P_{k+2,1}.
\end{equation}

\textbf{The $r^{th}$ round of repair ($r>1$)}. If the failed node in the
$r^{th}$ round of repair is the repaired node in the $(r-1)^{th}$ round of
repair, then we just repeat the $(r-1)^{th}$ repair. Otherwise, we 
select the $k+1$ chunks that are {\em different} from those selected
in the $(r-1)^{th}$ round of repair.  For example, in Figure~\ref{fig:fmsr},
if in the next round of repair the failed node remains node~1, then $P_{2,1}$,
$P_{3,1}$, and $P_{4,1}$, which have been selected in Figure~\ref{fig:fmsr},
are selected for the next repair. If the failed node is node~2, then we should
select $P_{1,1}$ (or $P_{1,2}$), $P_{3,2}$, and $P_{4,2}$.  
%Note that P_{3,2} is different from P_{3,1}, and P_{4,2} is different from
%P4,1; but either of P1,1 and P1,2 can be selected  for they have not been
%selected in Fig. 1
Then similar to the first round of
repair, we generate the coefficients that satisfy inequalities likewise
in (\ref{fml:deterministic_condition_1}),
(\ref{fml:deterministic_condition_2}), and 
(\ref{fml:deterministic_condition_3}). Finally, we regenerate the new chunks
accordingly as (\ref{fml:deterministic_combination_1}) and
(\ref{fml:deterministic_combination_2}).

\section{Evaluation}
\label{sec:evaluation}

% define coefficients

In this section, we evaluate the repair performance of two implementations of
FMSR codes: (i) {\em random FMSR codes}, which use random chunk selection in
repair and is used in NCCloud \cite{Hu12} and (ii) {\em deterministic FMSR
codes}, which use deterministic chunk selection proposed in
Section~\ref{sec:deterministic}.  We show that our proposed deterministic FMSR
codes can significantly reduce the time required to regenerate parity chunks
in repair.

% Methodology
We implement both versions of FMSR codes in C. We implement finite-field
arithmetic operations over a Galois Field GF($2^8$) based on the standard
table lookup approach \cite{Greenan08}.  We conduct our evaluation on a
server running on an Intel CPU at 2.4GHz.  We consider different values of $n$
(i.e., the number of nodes).  For each $n$, we first apply Reed-Solomon codes
to generate the encoding coefficients that will be used to encode a file into
parity chunks before uploading.  In each round of repair, we randomly pick a
node to fail. We then repair the failed node using two-phase checking, based
on either random or deterministic FMSR code implementations.  The failed node
that we choose is different from that of the previous round of repair, so as
to ensure a different chunk selection in each round of repair.  We conduct 50
rounds of repair in each evaluation run.  We conduct a total of 30 runs over
different seeds for each $n$.

The metric we are interested in is the checking time spent on determining if
the chunks selected from surviving nodes can be used to regenerate the lost
chunks.  We do not measure the times of reading or writing chunks, as they are
the same for both random and deterministic FMSR codes. Instead, we focus on
measuring the processing time of two-phase checking in each round of repair.
It is important to note that two-phase checking only operates on encoding
coefficients, and is independent of the size of the file being encoded.  Note
that we do not specifically optimize our encoding operations, but we believe
our results provide fair comparison of both random and deterministic FMSR
codes using our baseline implementations.

% Results
Figure~\ref{fig:agg_time} first depicts the aggregate checking times for a
total of 50 rounds of repair versus the number of nodes when using random and
deterministic FMSR codes.  The aggregate checking time of random FMSR codes is
small when $n$ is small (e.g., less than 1~second for $n\le 6$), but
exponentially increases as $n$ is large.  On the other hand, the aggregate
checking time of deterministic FMSR codes is significantly small (e.g., within
0.2~seconds for $n\le 10$).

\begin{figure}[t]
\centering
\includegraphics[width=2.3in]{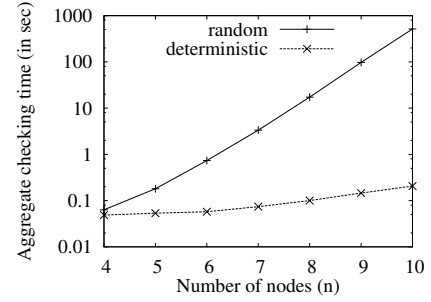}
\caption{Aggregate checking time of 50 rounds of repair (y-axis is in log
scale).}
\label{fig:agg_time}
\vspace{-1em}
\end{figure}

Our investigation finds that the checking time of random FMSR codes increases
dramatically as the value of $n$ increases. For example, when $n=12$ (not
shown in our figures), we find that the repair operation of our random FMSR
code implementation still cannot return a right set of regenerated chunks
after running for two hours.  In contrast, our deterministic FMSR codes can 
return a solution within 0.5~seconds. 

To further examine the significant performance overhead of random FMSR codes,
Figures~\ref{fig:inst_time} and \ref{fig:inst_iter} show the cumulative
checking time and number of two-phase checkings performed for $r$ rounds of
repair, respectively, for $n=8,9,10$.  We note that random FMSR codes incur
a fairly large but constant number of two-phase checkings in each round of
repair.  For example, for $n=10$, each round of repair takes around 100
iterations of two-phase checkings (see Figure~\ref{fig:inst_iter}(a)).  On the
other hand, deterministic FMSR codes significantly reduce the number of
iterations of two-phase checking (e.g., less than 2.5 on average for $n=10$).
{\em In summary, our evaluation results show that deterministic FMSR codes
significantly reduce the two-phase checking overhead of ensuring that the MDS
property is preserved during repair.}

\begin{figure}[t]
\centering
\begin{tabular}{c@{\ }c}
\includegraphics[width=1.7in]{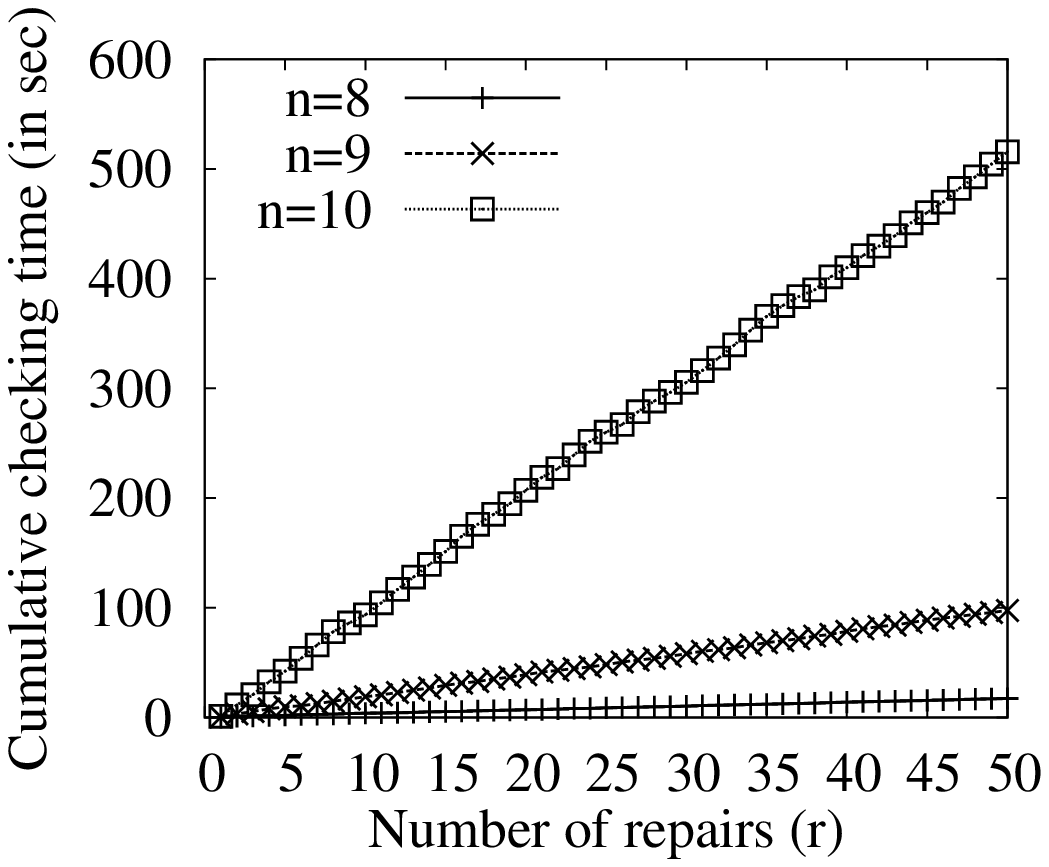} &
\includegraphics[width=1.7in]{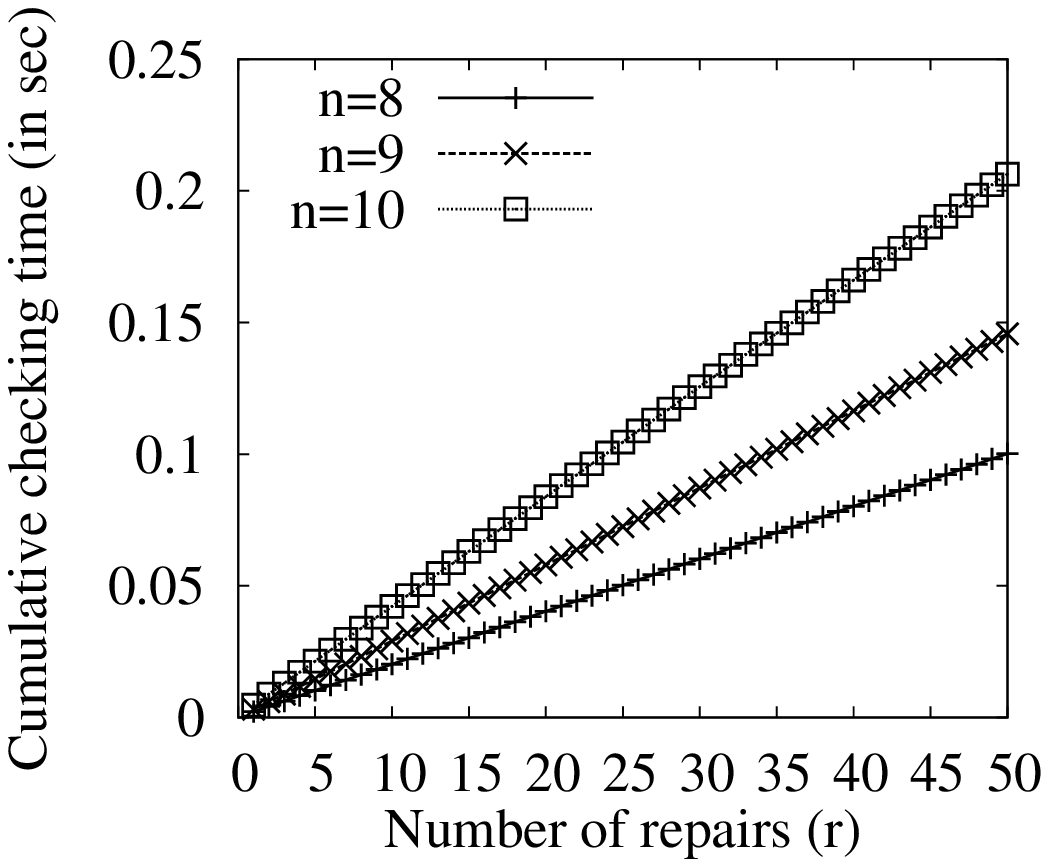} \\
\mbox{(a) random FMSR codes} &
\mbox{(b) deterministic FMSR codes}
\end{tabular}
\caption{Cumulative checking time of $r$ rounds of repair.}
\label{fig:inst_time}
\vspace{-1em}
\end{figure}
\begin{figure}[t]
\centering
\begin{tabular}{c@{\ }c}
\includegraphics[width=1.7in]{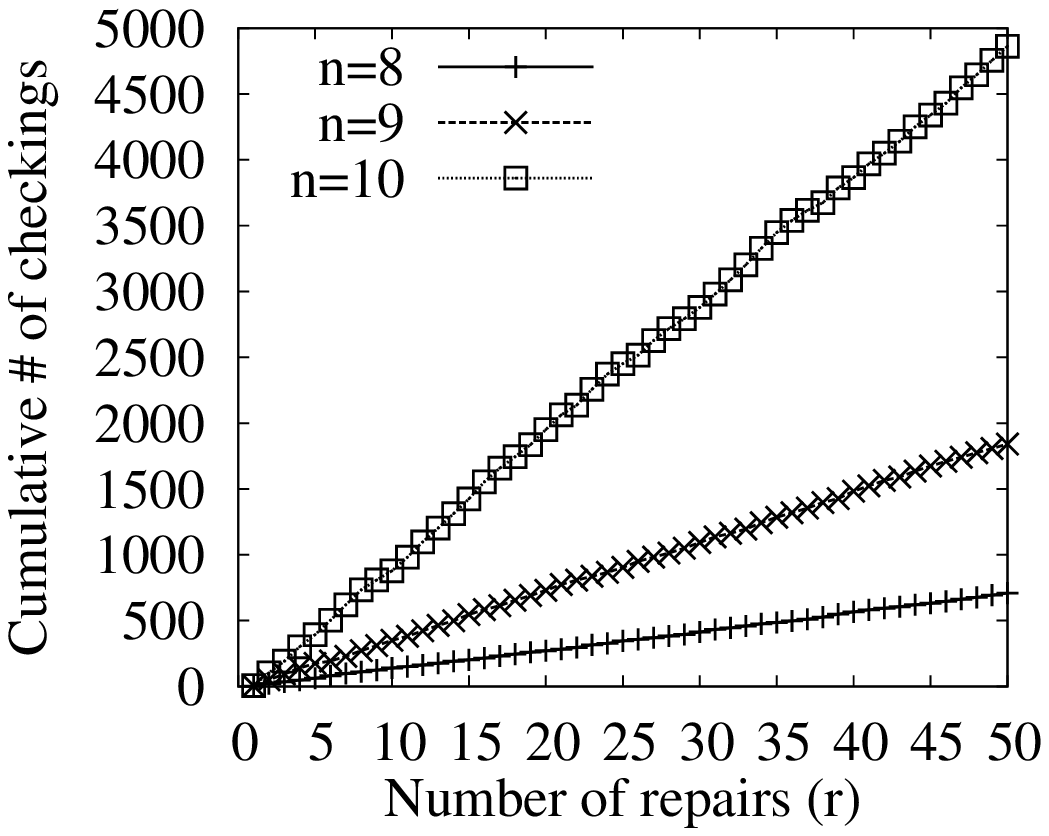} &
\includegraphics[width=1.7in]{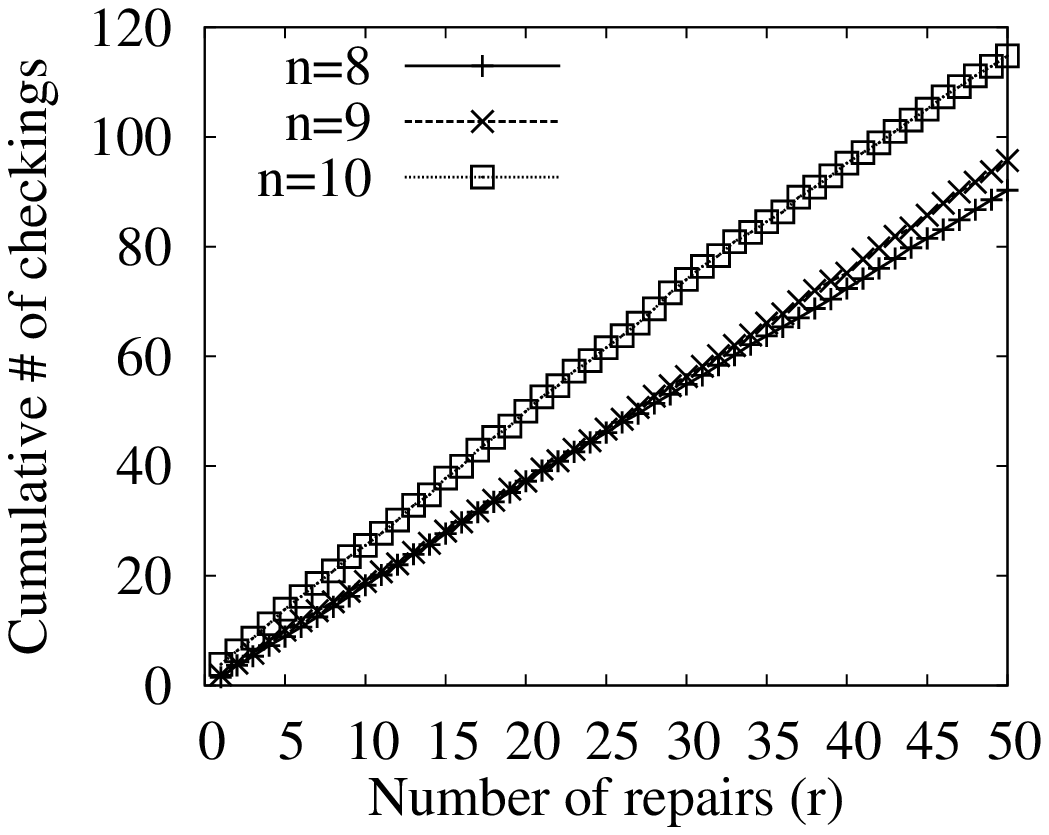} \\
\mbox{(a) random FMSR codes} &
\mbox{(b) deterministic FMSR codes}
\end{tabular}
\caption{Cumulative number of two-phase checkings of $r$ rounds of repair.}
\label{fig:inst_iter}
\vspace{-1em}
\end{figure}

\section{Conclusions}
\label{sec:conclusions}

This paper formulates an uncoded repair problem based on functional minimum
storage regenerating (FMSR) codes. We formally prove the existence of FMSR
codes and provide a deterministic FMSR code construction.  We also show via
our evaluation that our deterministic FMSR codes significantly reduce the
repair time overhead of random FMSR codes.  Our theoretical results validate
the correctness of existing practical FMSR code implementation \cite{Hu12}. We
also demonstrate the feasibility of preserving the benefits of network coding
in minimizing the repair bandwidth with uncoded repair. 

%Our future work is to design a more generalized FMSR codes for any $k$ and
%$n$.

\section*{Acknowledgments}

This work is supported by grants AoE/E-02/08 and ECS CUHK419212 from the
University Grants Committee of Hong Kong.

%----------------------------------------------------------------------------
\bibliographystyle{abbrv}
\bibliography{paper}

\begin{thebibliography}{10}

\bibitem{ahlswede00}
R.~Ahlswede, N.~Cai, S.-Y.~R. Li, and R.~W. Yeung.
\newblock {Network Information Flow}.
\newblock {\em IEEE Trans. on Info. Theory}, 46(4):1204--1216, Jul 2000.

\bibitem{Bhagwan04}
R.~Bhagwan, K.~Tati, Y.~Cheng, S.~Savage, and G.~Voelker.
\newblock {Total Recall: System Support for Automated Availability Management}.
\newblock In {\em Proc. of NSDI}, 2004.

\bibitem{Cadambe10}
V.~R. Cadambe, S.~A. Jafar, and H.~Maleki.
\newblock Distributed data storage with minimum storage regenerating codes -
  exact and functional repair are asymptotically equally efficient.
\newblock ar{X}iv:1004.4299 [cs.IT], 2010.

\bibitem{Calder11}
B.~Calder~et al.
\newblock {Windows Azure Storage: A Highly Available Cloud Storage Service with
  Strong Consistency}.
\newblock In {\em Proc. of ACM SOSP}, 2011.

\bibitem{Chun06}
B.~Chun, F.~Dabek, A.~Haeberlen, E.~Sit, H.~Weatherspoon, M.~F. Kaashoek,
  J.~Kubiatowicz, and R.~Morris.
\newblock {Efficient Replica Maintenance for Distributed Storage Systems}.
\newblock In {\em Proc. of NSDI}, 2006.

\bibitem{Cullina09}
D.~Cullina, A.~G. Dimakis, and T.~Ho.
\newblock Searching for minimum storage regenerating codes.
\newblock In {\em Proc. of Allerton}, 2009.

\bibitem{dimakis10}
A.~G. Dimakis, P.~B. Godfrey, Y.~Wu, M.~Wainwright, and K.~Ramchandran.
\newblock {Network Coding for Distributed Storage Systems}.
\newblock {\em IEEE Trans. on Info. Theory}, 56(9):4539--4551, Sep 2010.

\bibitem{Ghemawat03}
S.~Ghemawat, H.~Gobioff, and S.-T. Leung.
\newblock {The Google File System}.
\newblock In {\em Proc. of ACM SOSP}, 2003.

\bibitem{Greenan08}
K.~M. Greenan, E.~L. Miller, and T.~J.~E. Schwarz.
\newblock {Optimizing Galois Field Arithmetic for Diverse Processor
  Architectures and Applications}.
\newblock In {\em Proc. of IEEE MASCOTS}, 2008.

\bibitem{Hu12}
Y.~Hu, H.~C.~H. Chen, P.~P.~C. Lee, and Y.~Tang.
\newblock {NCCloud: Applying Network Coding for the Storage Repair in a
  Cloud-of-Clouds}.
\newblock In {\em Proc. of FAST}, 2012.

\bibitem{hu10}
Y.~Hu, Y.~Xu, X.~Wang, C.~Zhan, and P.~Li.
\newblock {Cooperative Recovery of Distributed Storage Systems from Multiple
  Losses with Network Coding}.
\newblock {\em IEEE JSAC}, 28(2):268--276, Feb 2010.

\bibitem{raid6}
Intel.
\newblock {Intelligent RAID6 Theory Overview and Implementation}, 2005.

\bibitem{khan12}
O.~Khan, R.~Burns, J.~S. Plank, W.~Pierce, and C.~Huang.
\newblock {Rethinking Erasure Codes for Cloud File Systems: Minimizing I/O for
  Recovery and Degraded Reads}.
\newblock In {\em Proc. of USENIX FAST}, 2012.

\bibitem{Kubiatowicz00}
J.~Kubiatowicz, D.~Bindel, Y.~Chen, S.~Czerwinski, P.~Eaton, D.~Geels,
  R.~Gummadi, S.~Rhea, H.~Weatherspoon, W.~Weimer, C.~Wells, and B.~Zhao.
\newblock Ocean{S}tore: an architecture for global-scale persistent storage.
\newblock In {\em Proc. of ASPLOS}, 2000.

\bibitem{SZtheorem}
R.~Motwani and P.~Raghavan.
\newblock Randomized algorithms.
\newblock In {\em Cambridge University Press}, 1995.

\bibitem{Plank97}
J.~S. Plank.
\newblock {A Tutorial on Reed-Solomon Coding for Fault-Tolerance in RAID-like
  Systems}.
\newblock {\em Software - Practice \& Experience}, 27(9):995--1012, Sep 1997.

\bibitem{rashmi11}
K.~V. Rashmi, N.~B. Shah, and P.~V. Kumar.
\newblock {Optimal exact-regenerating codes for distributed storage at the MSR
  and MBR points via a product-matrix construction}.
\newblock {\em IEEE Trans. on Info. Theory}, 57(8):5227--5239, Aug. 2011.

\bibitem{reed60}
I.~Reed and G.~Solomon.
\newblock {Polynomial Codes over Certain Finite Fields}.
\newblock {\em Journal of the Society for Industrial and Applied Mathematics},
  8(2):300--304, 1960.

\bibitem{Rouayheb10}
S.~E. Rouayheb and K.~Ramchandran.
\newblock {Fractional Repetition Codes for Repair in Distributed Storage
  Systems}.
\newblock In {\em Proc. of Allerton}, 2010.

\bibitem{Shah12b}
N.~B. Shah, K.~V. Rashmi, P.~V. Kumar, and K.~Ramchandran.
\newblock {Distributed storage codes with Repair-by-Transfer and
  Non-Achievability of Interior Points on the Storage-Bandwidth Tradeoff}.
\newblock {\em IEEE Trans. on Info. Theory}, 58(3):1837--1852, Mar. 2012.

\bibitem{shah12}
N.~B. Shah, K.~V. Rashmi, P.~V. Kumar, and K.~Ramchandran.
\newblock Interference alignment in regenerating codes for distributed storage:
  Necessity and code constructions.
\newblock {\em IEEE Trans. on Info. Theory}, 58(4):2134--2158, Sept. 2012.

\bibitem{Shum12}
K.~W. Shum and Y.~Hu.
\newblock Functional-repair-by-transfer regenerating code.
\newblock In {\em Proc. of ISIT}, 2012.

\bibitem{suh11b}
C.~Suh and K.~Ramchandran.
\newblock {Exact-Repair MDS Code Construction Using Interference Alignment}.
\newblock {\em IEEE Trans. on Info. Theory}, 57(3):1425--1442, Mar. 2011.

\bibitem{tamo11}
I.~Tamo, Z.~Wang, and J.~Bruck.
\newblock {MDS Array Codes with Optimal Rebuilding}.
\newblock In {\em Proc. of ISIT}, 2011.

\bibitem{wang10}
Z.~Wang, A.~Dimakis, and J.~Bruck.
\newblock {Rebuilding for Array Codes in Distributed Storage Systems}.
\newblock In {\em IEEE GLOBECOM Workshops}, 2010.

\bibitem{wang11}
Z.~Wang, I.~Tamo, and J.~Bruck.
\newblock {On Codes for Optimal Rebuilding Access}.
\newblock In {\em Proc. of Allerton}, 2011.

\bibitem{Weatherspoon07}
H.~Weatherspoon, P.~Eaton, B.~Chun, and J.~Kubiatowicz.
\newblock {Antiquity: Exploiting a Secure log for Wide-Area Distributed
  Storage}.
\newblock In {\em Proc. of ACM SIGOPS/EuroSys}, 2007.

\bibitem{Yunnanwu10}
Y.~Wu.
\newblock {Existence and Construction of Capacity-Achieving Network Codes for
  Distributed Storage}.
\newblock {\em IEEE JSAC}, 28(2), Feb. 2010.

\bibitem{wu09a}
Y.~Wu and A.~G. Dimakis.
\newblock Reducing repair traffic for erasure coding-based storage via
  interference alignment.
\newblock In {\em Proc. of ISIT}, 2009.

\bibitem{xiang11}
L.~Xiang, Y.~Xu, J.~Lui, Q.~Chang, Y.~Pan, and R.~Li.
\newblock {A Hybrid Approach to Failed Disk Recovery Using RAID-6 Codes:
  Algorithms and Performance Evaluation}.
\newblock {\em ACM Trans. on Storage}, 7(3):11, 2011.

\end{thebibliography}

%\appendix[Proof of Correctness of Deterministic FMSR Codes]
\appendix

We now prove the correctness of the deterministic FMSR codes in
Section~\ref{sec:deterministic}.  Initiaially, the file is stored with
Reed-Solomon codes, such that any $2k$ out of $2(k+2)$ (parity) chunks are
decodable to the original file.  Therefore, the set of chunks being stored
before any repair satisfies the MDS and erMDS properties.  Now, we show that
the MDS and erMDS properties are always satisfied after each round of repair,
based on our chunk selection and coefficient construction.

\textbf{The first round of repair.} Let $\SU_1=\{P'_{1,1},P'_{1,2};$
$P_{2,1},P_{2,2};$ $\ldots;$ $P_{n,1},P_{n,2}\}$ be the set of all chunks
after the first round of repair (for failed node~1).  Next we prove that
$\SU_1$ still satisfies both the MDS and erMDS properties.

(\emph{$\SU_1$ satisfies the MDS property}) Since the file is stored with
Reed-Solomon Codes, all the chunks of any $k$ out of nodes $2,\ldots,k+2$
before the repair are obviously decodable. Thus, we only need to check whether
the chunks of the repaired node $1$ and any $k-1$ of nodes $2,\ldots,k+2$ are
decodable. Take the repaired node~1 and nodes $2,\ldots,k$ for instance. 
Denote the $2k$ chunks of them by $\SV =$ $\{P'_{1,1},P'_{1,2};$ 
$P_{2,1},P_{2,2};$ $\ldots;$ $P_{k,1},P_{k,2}\}$. Consider the {\em linear
span} of $\SV$ (i.e., the set of all linear combinations of $\SV$). 
Due to Equations~(\ref{fml:deterministic_combination_1}) and
(\ref{fml:deterministic_combination_2}), the linear span of $\SV$ can be
expressed as span($\SV$) $=$
span($\gamma_{k+1,1} P_{k+1,1} +\gamma_{k+2,1} P_{k+2,1}, \gamma_{k+1,2} P_{k+1,1}+\gamma_{k+2,2} P_{k+2,1}$;$P_{2,1},P_{2,2};$ $\ldots;$
$P_{k,1},P_{k,2} $).  Note that the coefficients are chosen in a way such that 
$\gamma_{k+1,1} \gamma_{k+2,2} \neq \gamma_{k+1,2} \gamma_{k+2,1} $ is
satisfied, based on inequality (\ref{fml:deterministic_condition_1}). 
So span($\SV$) $=$ span($P_{k+1,1},P_{k+2,1}$;$P_{2,1},P_{2,2};$ $\ldots;$$P_{k,1},P_{k,2} $).
Based on the erMDS property, $\SV$ is decodable because its linear span 
contains $P_{2,1},P_{3,1},\ldots, P_{k+2,1}$ from nodes $2,\ldots,k+2$,
respectively.

(\emph{$\SU_1$ satisfies the erMDS property})
Since the file is initially stored with Reed-Solomon Codes, the erMDS property
is satisfied before the repair. Hence there already exist $k+1$ chunks, say
$P_{2,1},\ldots,P_{k+2,1}$, such that any RBC containing them is decodable. 
Thus, we only need to check whether for the repaired node~1 and any $k$ of
nodes $2,\ldots,k+2$, there always exist $k+1$ chunks such that by collecting
one chunk from each such node, any RBC containing them is decodable.  Without
loss of generality, we just consider the case for the repaired node~1 and
nodes~$2,\ldots,k+1$ for simplicity.

Here, we select the $k+1$ chunks in the way that they are {\em distinct} from
those selected for the first round of repair.  In this case, we collect $\SF_1
= \{P'_{1,2},P_{2,2},\ldots,P_{k+1,2}\}$ (note: either $P'_{1,1}$ or
$P'_{1,2}$ is fine).  Next we show that the constructed $\gamma_{i,j}$ can
make any RBC containing $\SF_1$ decodable.  Since the repaired node~1 may
offer one or two chunks to an RBC, we consider two cases.

\emph{Case 1}: The repaired node~1 only offers one chunk. Then the RBC needs
another $k-1$ nodes (e.g., nodes~$2,\ldots,k$) to offer all their chunks and
another one node (e.g., node~$k+1$) to offer one chunk.  To make the RBC
include $\SF_1$, we have the repaired node~1 offer $P'_{1,2}$ and node~$k+1$
offer $P_{k+1,2}$.  Then the RBC is
$\SR_1 = \{P'_{1,2};P_{2,1},P_{2,2};\ldots;P_{k,1},
P_{k,2};P_{k+1,2}\}$.  By Equation~(\ref{fml:deterministic_combination_2}),
span($\SR_1$) $=$ span($\gamma_{k+1,2}P_{k+1,1}+\gamma_{k+2,2}P_{k+2,1};$
$P_{2,1},P_{2,2};\ldots; P_{k,1},P_{k,2};P_{k+1,2}$).

Based on the MDS property, we consider a decodable collection
$\SZ = \{P_{2,1},P_{2,2};\ldots;P_{k+1,1},P_{k+1,2}\}$. Then $P_{k+2,1}$ is a
linear combination of $\SZ$, and can be expressed as
\begin{equation}
\label{equ:appendix}
P_{k+2,1}=\sum_{i=2}^{k+1} \sum_{j=1}^2 \lambda_{i,j}^{(k+2)} P_{i,j},
%(\lambda_{2,1}^{k+2}, \lambda_{2,2}^{k+2},\ldots, \lambda_{k+1,1}^{k+2},\lambda_{k+1,2}^{k+2})[{\cal{Z}}]^T,
\end{equation}
where $\lambda_{i,j}^{(k+2)}$ is an encoding coefficient for
$i=2,\ldots,k+1$ and $j=1,2$.  Thus, the linear span of $\SR_1$ is
span($\SR_1$) $=$
span($(\gamma_{k+1,2}+\gamma_{k+2,2}\lambda_{k+1,1}^{(k+2)})P_{k+1,1};$
$P_{2,1},P_{2,2};$ $\ldots;P_{k,1},P_{k,2};P_{k+1,2}\}$).  Note that the
coefficients are chosen such that 
$\gamma_{k+1,2}+\gamma_{k+2,2}\lambda_{k+1,1}^{(k+2)} \neq 0$ is satisfied,
based on inequality (\ref{fml:deterministic_condition_2}). 
Thus, span($\SR_1$) $=$ span($P_{2,1},P_{2,2};\ldots;P_{k,1},P_{k,2}$). The
linear span of $\SR_1$ is a decodable collection due to the MDS property.
Thus, $\SR_1$ is decodable. 

\emph{Case 2}: The repaired node~1 offers two chunks. So the RBC contains both
$P'_{1,1}$ and $P'_{1,2}$.  The RBC needs another $k-2$ nodes (e.g.,
nodes~$2,\ldots,k-1$) to offer all their chunks and another two nodes (e.g.,
nodes $k$ and $k+1$) to offer one chunk.  To make the RBC contain $\SF_1$, we
have nodes~$k$ and $k+1$ offer $P_{k,2}$ and
$P_{k+1,2}$, respectively.  Then the RBC is $\SR_2 = \{P'_{1,1},P'_{1,2};$
$P_{2,1},P_{2,2};\ldots;$ $P_{k-1,1},P_{k-1,2};P_{k,2};P_{k+1,2}\}$. 
%
%By Equations (\ref{fml:deterministic_combination_1}) and
%(\ref{fml:deterministic_combination_2}), span($\SR_2$) $=$ span($\gamma_{k,1}
%P_{k,1} +\gamma_{k+1,1} P_{k+1,1} +\gamma_{k+2,1} P_{k+2,1}$, $\gamma_{k,1}
%P_{k,1} +\gamma_{k+1,1} P_{k+1,1} +\gamma_{k+2,1}
%P_{k+2,1}$,$P_{2,1},P_{2,2};\ldots;P_{k-1,1},P_{k-1,2};P_{k,2};P_{k+1,2}$). 
%By Equation~\ref{equ:appendix}, 
Similar to the proof of Case~1, by Equations~
(\ref{fml:deterministic_combination_1}), 
(\ref{fml:deterministic_combination_2}), and
(\ref{equ:appendix}), the linear span of $\SR_2$ can be expressed as 
\[
\begin{array}{ll}
\textrm{span}(\SR_2) = & 
\textrm{span}((\gamma_{k,1} + \gamma_{k+2,1}\lambda_{k,1}^{(k+2)})P_{k,1} + \\
		& \ \ (\gamma_{k+1,1}+\gamma_{k+2,1}\lambda_{k+1,1}^{(k+2)})P_{k+1,1},\\
		&  (\gamma_{k,2} + \gamma_{k+2,2}\lambda_{k,1}^{(k+2)})P_{k,1} + \\
		& \ \ (\gamma_{k+1,2}+\gamma_{k+2,2}\lambda_{k+1,1}^{(k+2)})P_{k+1,1},\\
		& P_{2,1},P_{2,2};\ldots;P_{k-1,1},P_{k-1,2};P_{k,2};P_{k+1,2}).
\end{array}
\]
Note that the coefficients are chosen in a way such that 
$(\gamma_{k,1}+\gamma_{k+2,1}\lambda_{k,1}^{(k+2)})(\gamma_{k+1,2}+\gamma_{k+2,2}\lambda_{k+1,1}^{(k+2)})
\neq
(\gamma_{k+1,1}+\gamma_{k+2,1}\lambda_{k+1,1}^{(k+2)})(\gamma_{k,2}+\gamma_{k+2,2}\lambda_{k,1}^{(k+2)})$
is satisfied, based on inequality (\ref{fml:deterministic_condition_3}). 
Thus, span($\SR_2$) $=$ 
span($\{P_{2,1},P_{2,2};\ldots;P_{k+1,1},P_{k+1,2}\}$). The linear span of
$\SR_2$ is decodable due to the MDS property. Thus, $\SR_2$ is decodable. 

\textbf{The $r^{th}$ repair ($r>1$)}
Take $r=2$ for instance. Suppose without loss of generality that node $k+2$
fails. Then we select $\{P'_{1,2},P_{2,2},\ldots,P_{k+1,2}\}$ which are
distinct from those in the first round of repair. We can observe that in fact
this set is $\SF_1$ in the first round of repair. As mentioned above, any RBC
containing $\SF_1$ is decodable. So $\SF_1$ can be used for the second round
of repair. Then we can generate the coefficients that satisfy the similar
inequalities
as~(\ref{fml:deterministic_condition_1}),~(\ref{fml:deterministic_condition_2}),
and~(\ref{fml:deterministic_condition_3}). The proof of correctness is similar
as $r=1$ and thus omitted.

\end{document}